\documentclass[11pt]{article}

\usepackage{color}
\usepackage{amsmath,amsfonts,amsthm}
\usepackage{hyperref}
\usepackage{mathtools}
\usepackage{lmodern}
\usepackage{multirow}
\usepackage{array}

\usepackage{fullpage}
\usepackage{amsmath,amsthm,amsfonts,dsfont}
\usepackage{amssymb,latexsym,graphicx}
\usepackage{mathtools}
\usepackage{hyperref}
\usepackage{xcolor}
\usepackage{microtype}
\hypersetup{
	colorlinks,
	linkcolor={red!75!black},
	citecolor={blue!75!black},
	urlcolor={blue!75!black}
}

\newtheorem{theorem}{Theorem}[section]

\newtheorem{lemma}[theorem]{Lemma}

\newtheorem{corollary}[theorem]{Corollary}
\newtheorem{definition}[theorem]{Definition}
\newtheorem{fact}[theorem]{Fact}
\newtheorem{proposition}[theorem]{Proposition}

\newcommand{\N}{\ensuremath{\mathbb{N}}}

\newcommand{\R}{\ensuremath{\mathbb{R}}}
\newcommand{\Z}{\ensuremath{\mathbb{Z}}}

\newcommand{\lat}{\mathcal{L}}
\newcommand{\M}{\mathcal{M}}

\renewcommand{\epsilon}{\varepsilon}
\newcommand{\poly}{\mathrm{poly}}

\DeclareMathOperator*{\argmin}{arg\,min}

\DeclareMathOperator{\dist}{dist}

\renewcommand{\vec}[1]{\ensuremath{\boldsymbol{#1}}}
\newcommand{\basis}{\ensuremath{\mathbf{B}}}
\newcommand{\gs}[1]{\ensuremath{\widetilde{\vec{#1}}}}

\DeclarePairedDelimiter\floor{\lfloor}{\rfloor}
\DeclarePairedDelimiter\length{\lVert}{\rVert}
\usepackage[boxed]{algorithm2e}

\newcommand{\pmax}{{p}_{\mathsf{max}}}
\newcommand{\pcol}{{p}_{\mathsf{col}}}

\title{
	Just Take the Average!\protect\\
	An Embarrassingly Simple $2^n$-Time Algorithm for SVP (and CVP)}

	\author{Divesh Aggarwal\\
		Centre for Quantum Technologies\\
		\texttt{dcsdiva@nus.edu.sg}
		\and Noah Stephens-Davidowitz\thanks{Supported by the National Science Foundation (NSF) under Grant No.~CCF-1320188, and the Defense Advanced Research Projects Agency (DARPA) and Army Research
			Office (ARO) under Contract No.~W911NF-15-C-0236.}\\
		New York University\\ 
		\texttt{noahsd@gmail.com}}

\date{}

\SetAlCapSkip{1em}

\begin{document}

\maketitle

\begin{abstract}
	We show a $2^{n+o(n)}$-time (and space) algorithm for the Shortest Vector Problem on lattices (SVP) that works by repeatedly running an embarrassingly simple ``pair and average'' sieving-like procedure on a list of lattice vectors. This matches the running time (and space) of the current fastest known algorithm, due to Aggarwal, Dadush, Regev, and Stephens-Davidowitz (ADRS, in \emph{STOC}, 2015), with a far simpler algorithm. Our algorithm is in fact a modification of the ADRS algorithm, with a certain careful rejection sampling step removed.
	
	The correctness of our algorithm follows from a more general ``meta-theorem,'' showing that such rejection sampling steps are unnecessary for a certain class of algorithms and use cases. In particular, this also applies to the related $2^{n + o(n)}$-time algorithm for the Closest Vector Problem (CVP), due to Aggarwal, Dadush, and Stephens-Davidowitz (ADS, in \emph{FOCS}, 2015), yielding a similar embarrassingly simple algorithm for $\gamma$-approximate CVP for any $\gamma = 1+2^{-o(n/\log n)}$. (We can also remove the rejection sampling procedure from the $2^{n+o(n)}$-time ADS algorithm for \emph{exact} CVP, but the resulting algorithm is still quite complicated.)
\end{abstract}

\section{Introduction}

A lattice $\lat \subset \R^n$ is the set of all integer linear combinations of some
linearly independent basis vectors $\vec{b}_1,\dots,\vec{b}_n \in \R^n$,
\[
\lat := \Big\{ \sum_{i=1}^n z_i \vec{b}_i \ : \ z_i \in \Z \Big\}
\; .
\]

The two most important computational problems on lattices are the Shortest Vector Problem (SVP) and the Closest Vector Problem (CVP). Given a basis $\vec{b}_1,\dots,\vec{b}_n$ for a lattice $\lat \subset \R^n$,
SVP asks us to find a shortest non-zero vector in $\lat$, and CVP asks us to find a closest lattice vector to some target vector $\vec{t} \in \R^n$. (Throughout this paper, we define distance in terms of the Euclidean, or $\ell_2$, norm.) CVP seems to be the harder of the two problems, as there is an efficient reduction from CVP to SVP that preserves the dimension $n$~\cite{GMSS99}, but both problems are known to be NP-hard~\cite{Boas81,Ajtai-SVP-hard}. They are even known to be hard to approximate for certain approximation factors~\cite{Mic01svp,DKRS03,Khot05svp,HRsvp}.

Algorithms for solving these problems, both exactly and over a wide range of approximation factors, have found innumerable applications since the founding work by Lenstra, Lenstra, and Lov{\'a}sz in 1982~\cite{LLL82}. (E.g.,~\cite{LLL82,Lenstra83,Shamir84,Kannan87,Buda89}.) More recently, following the celebrated work of Ajtai~\cite{Ajtai96} and Regev~\cite{oded05}, a long series of works has resulted in many cryptographic constructions whose security is based on the assumed \emph{worst-case} hardness of approximating these (or closely related) problems. (See~\cite{chris_survey} for a survey of such constructions.) And, some of these constructions are now nearing widespread deployment. (See, e.g., \cite{NIST_quantum,new_hope,frodo}.)

Nearly all of the fastest known algorithms for lattice problems---either approximate or exact---work via a reduction to either exact SVP or exact CVP (typically in a lower dimension). Even the fastest known polynomial-time algorithms (which solve lattice problems only up to large approximation factors) work by solving exact SVP on low-dimensional sublattices~\cite{Schnorr87,GN08,MW16}. Therefore, algorithms for exact lattice problems are of particular importance, both theoretically and practically (and both for direct applications and to aid in the selection of parameters for cryptography). Indeed, much work has gone into improving the running time of these algorithms (e.g.,~\cite{Kannan87,AKS01,AKS02,PS09,MV10,MV13}), culminating in $2^{n +o(n)}$-time algorithms for both problems based on the technique of \emph{discrete Gaussian sampling}, from joint work with Dadush and Regev~\cite{ADRS15} and follow-up work with Dadush~\cite{ADS15}.

In order to explain our contribution, we first give a high-level description of the SVP algorithm from~\cite{ADRS15}. (The presentation below does not represent the way that we typically view the~\cite{ADRS15} algorithm.)

\subsection{Sieving by averages}

One can think of the SVP algorithm from~\cite{ADRS15} as a rather strange variant of \emph{randomized sieving}. Recall that the celebrated randomized sieving technique due to Ajtai, Kumar, and Sivakumar~\cite{AKS01} starts out with a list of $2^{O(n)}$ not-too-long random vectors $\vec{X}_1,\ldots, \vec{X}_M$ sampled from some efficiently samplable distribution. The sieving algorithm then repeatedly (1) searches for pairs of vectors $(\vec{X}_i, \vec{X}_j)$ that happen to be remarkably close together; and then (2) replaces the old list of vectors with the differences of these pairs $\vec{X}_i - \vec{X}_j$.

The~\cite{ADRS15} algorithm similarly starts with a randomly chosen collection of $2^{n + o(n)}$ not-too-long vectors $\vec{X}_1,\ldots, \vec{X}_M$ and repeatedly (1) selects pairs of vectors according to some rule; and (2) replaces the old list of vectors with some new vectors generated from these pairs. However, instead of taking the differences $\vec{X}_i - \vec{X}_j$ of pairs $(\vec{X}_i, \vec{X}_j)$, the~\cite{ADRS15} algorithm takes \emph{averages}, $(\vec{X}_i + \vec{X}_j)/2$.

Notice that the average $(\vec{X}_i + \vec{X}_j)/2$ of two lattice vectors $\vec{X}_i, \vec{X}_j \in \lat$ will not generally be in the lattice. In fact, this average will be in the lattice if and only if the two vectors are equivalent mod $2\lat$, i.e., $\vec{X}_i \equiv \vec{X}_j \bmod 2\lat$. Therefore, at a minimum, the~\cite{ADRS15} algorithm should select pairs that lie in the same \emph{coset} mod $2\lat$. (Notice that there are $2^n$ possible cosets.) I.e., the simplest possible version of ``sieving by averages'' just repeats Procedure~\ref{proc:pair_and_average} many times (starting with a list of $2^{n+o(n)}$ vectors, which is sufficient to guarantee that we can pair nearly every vector with a different unique vector in the same coset). The~\cite{ADRS15} algorithm is more complicated than this, but it still only uses the cosets of the vectors mod $2\lat$ when it decides which vectors to pair.

\begin{algorithm}
	\SetKwInOut{Input}{Input}
	\SetKwInOut{Output}{Output}
	
	\underline{Pair\_and\_Average} $(\vec{X}_1,\ldots, \vec{X}_M)$\\
	\Input{\ \ List of vectors $\vec{X}_i \in \lat - \vec{t}$}
	\Output{\ \ List of vectors $\vec{Y}_i \in \lat - \vec{t}$}
	\For{each unpaired vector $\vec{X}_i$}
	{
	\If{there exists an unpaired vector $\vec{X}_j$ with $\vec{X}_j \equiv \vec{X}_i \bmod 2\lat$}{add $(\vec{X}_i + \vec{X}_j)/2$ to the output}	
	}
	\caption{\label{proc:pair_and_average}The basic ``pair and average'' procedure, which computes the averages $(\vec{X}_i + \vec{X}_j)/2$ of disjoint pairs $(\vec{X}_i, \vec{X}_j)$ satisfying $\vec{X}_i \equiv \vec{X}_j \bmod 2\lat$.}
\end{algorithm}

It might seem like a rather big sacrifice to only look at a vector's coset, essentially ignoring all geometric information. For example, the~\cite{ADRS15} algorithm (and our variant) is likely to miss many opportunities to pair two vectors whose average is very short.
But, in exchange for this sacrifice, we get very strong control over the distribution of the vectors at each step. In particular, before applying Procedure~\ref{proc:pair_and_average}, the~\cite{ADRS15} algorithm uses a careful \emph{rejection sampling} procedure over the cosets to guarantee that at each step of the algorithm, the vectors are distributed as independent samples from a distribution that we understand very well (the discrete Gaussian, which we describe in Section~\ref{sec:techniques}). I.e., at each step, the algorithm randomly throws away many of the vectors in each coset according to some rule that depends only on the list of cosets, and it only runs Procedure~\ref{proc:pair_and_average} on the remaining vectors, as shown in Procedure~\ref{proc:reject}. This rejection sampling procedure is so selective that, though the algorithm starts out with $2^{n + o(n)}$ vectors, it typically finishes with only about $2^{n/2}$ vectors.

\begin{algorithm}
	\SetKwInOut{Input}{Input}
	\SetKwInOut{Output}{Output}
	
	\underline{Reject-and-Average Sieve} $(\ell; \vec{X}_1,\ldots, \vec{X}_M)$\\
	\Input{\ \ Number of steps $\ell$, list of vectors $\vec{X}_i \in \lat - \vec{t}$}
	\Output{\ \ List of vectors $\vec{Y}_i \in \lat - \vec{t}$}
	\For{$i = 1,\ldots, \ell$}
	{
		\For{$j = 1,\ldots, M$}{set $\vec{c}_j$ to be the coset of $\vec{X}_j$ mod $2\lat$}
		$\{j_1,\ldots, j_m\} \leftarrow f(\vec{c}_1,\ldots, \vec{c}_M)$\\
		$(\vec{X}_1,\ldots, \vec{X}_{M'}) \leftarrow \text{Pair\_and\_Average}(\vec{X}_{j_1},\ldots, \vec{X}_{j_m})$\\
		$M \leftarrow M'$
	}
	output $(\vec{X}_1,\ldots, \vec{X}_{M})$.
	\caption{\label{proc:reject}The ``reject and average'' sieving procedure, which repeatedly applies some rejection sampling procedure $f$ according to the cosets of the $\vec{X}_i$ mod $2\lat$ and then applies Procedure~\ref{proc:pair_and_average} (Pair\_and\_Average) to the ``accepted'' vectors. Here, $f$ is some (possibly randomized) function that maps a list of cosets $(\vec{c}_1,\ldots, \vec{c}_M)$ mod $2\lat$ to a set of ``accepted'' indices $\{j_1,\ldots, j_m\} \subseteq \{1,\ldots, M\}$.}
\end{algorithm}

This does seem quite wasteful (since the algorithm typically throws away the vast majority of its input vectors) and a bit naive (since the algorithm ignores, e.g., the lengths of the vectors). But, because we get such good control over the output distribution, the result is still the fastest known algorithm for SVP.\footnote{There are various ``heuristic'' sieving algorithms for SVP that run significantly faster (e.g., in time $(3/2)^{n/2}$~\cite{BDGL16}) but do not have formal proofs of correctness. One of the reasons that these algorithms lack proofs is because we do not understand their output distributions.} (The~\cite{ADS15} algorithm for CVP relies on the same core idea, plus a rather complicated recursive procedure that converts an approximate CVP algorithm with certain special properties into an exact CVP algorithm.)

\subsection{Our contribution}

Our main contribution is to show that the rejection sampling procedure used in the~\cite{ADRS15} algorithm is unnecessary! Indeed, informally, we show that ``any collection of vectors that can be found via such a procedure (when the input vectors are sampled independently from an appropriate distribution) can also be found without it.'' (We make this precise in Theorem~\ref{thm:meta_thm}.) In particular, the SVP algorithm in~\cite{ADRS15} can be replaced by an extremely simple algorithm, which starts with a list of $2^{n +o(n)}$ vectors sampled from the right distribution and then just runs Procedure~\ref{proc:pair_and_average} repeatedly. (Equivalently, it runs Procedure~\ref{proc:reject} with $f$ taken to be the trivial function that always outputs all indices, $\{1,\ldots, M\}$.)

\begin{theorem}[SVP, informal]
	\label{thm:SVP_intro}
	There is a $2^{n + o(n)}$-time (and space) algorithm for SVP that starts with $2^{n + o(n)}$ vectors sampled from the same distribution as the~\cite{ADRS15} algorithm and then simply applies Procedure~\ref{proc:pair_and_average} repeatedly, $\ell = O(\log n)$ times.
\end{theorem}

The situation for CVP is, alas, more complicated because Procedure~\ref{proc:reject} is not the most difficult part of the \emph{exact} CVP algorithm from~\cite{ADS15}. Indeed, while this algorithm does run Procedure~\ref{proc:reject} and we \emph{do} show that we can remove the rejection sampling procedure, the resulting algorithm retains the complicated recursive structure of the original~\cite{ADS15} algorithm. However,~\cite{ADS15} also shows a much simpler non-recursive version of their algorithm that solves CVP up to an extremely good approximation factor. If we are willing to settle for such an algorithm, then we get the same result for CVP.

\begin{theorem}[CVP, informal]
	\label{thm:CVP_intro}
	There is a $2^{n + o(n)}$-time (and space) algorithm that approximates CVP up to an approximation factor $\gamma$  for any $\gamma = 1+2^{-o(n/\log n)}$ that starts with $2^{n + o(n)}$ vectors from the same distribution as the~\cite{ADS15} algorithm and then simply applies Procedure~\ref{proc:pair_and_average} repeatedly, $\ell = o(n/\log n)$ times.
\end{theorem}

\noindent In practice, such a tiny approximation factor is almost always good enough for applications. 

\subsection{Proof techniques}
\label{sec:techniques}

To describe the technical ideas behind our result, we now define the \emph{discrete Gaussian distribution}, which plays a fundamental role in the algorithms in~\cite{ADRS15,ADS15} and a big part in our analysis. For any vector $\vec{x} \in \R^n$ and parameter $s > 0$, we define its Gaussian mass as
\[
\rho_s(\vec{x}) := \exp(-\pi \|\vec{x}\|^2/s^2)
\; ,
\]
and we extend this definition to a shift of a lattice $\lat \subset \R^n$ with shift vector $\vec{t} \in \R^n$ in the natural way,
\[
\rho_s(\lat - \vec{t}) := \sum_{\vec{y} \in \lat}\rho_s(\vec{y} - \vec{t})
\; .
\]
The \emph{discrete Gaussian distribution} $D_{\lat - \vec{t}, s}$ is the probability distribution over $\lat - \vec{t}$ induced by this measure, given by
\[
\Pr_{\vec{X} \sim D_{\lat - \vec{t}, s}}[\vec{X} = \vec{y} - \vec{t}] := \frac{\rho_s(\vec{y} - \vec{t})}{\rho_s(\lat - \vec{t})}
\]
for any $\vec{y} \in \lat$.

For very large parameters $s> 0$, we can sample from the discrete Gaussian $D_{\lat - \vec{t}, s}$ efficiently~\cite{GPV08,BLPRS13}. (Notice that $D_{\lat - \vec{t}, s}$ tends to concentrate on shorter vectors as the parameter $s > 0$ gets smaller. In particular,~\cite{ADRS15} showed that about $1.38^n$ independent samples from the discrete Gaussian $D_{\lat, s}$ with an appropriately chosen parameter $s$ will contain a shortest non-zero lattice vector with high probability. See Proposition~\ref{prop:SVPtoDGS}.) So, in~\cite{ADRS15,ADS15}, we use Procedure~\ref{proc:reject} with a carefully chosen rejection sampling procedure $f$ in order to convert many independent samples from $D_{\lat - \vec{t}, s}$ with a relatively large parameter $s$ to some smaller number of independent samples from $D_{\lat - \vec{t}, s/2^{\ell/2}}$.

This rejection sampling is certainly necessary if we wish to use Procedure~\ref{proc:pair_and_average} to sample from the discrete Gaussian distribution. Our new observation is that, even when we do not do this rejection sampling, the output of Procedure~\ref{proc:pair_and_average} still has a nice distribution. In particular, if we fix the coset mod $2\lat$ of a pair of discrete Gaussian vectors $(\vec{X}_i, \vec{X}_j)$ with parameter $s > 0$, then their average will be distributed as a \emph{mixture} of discrete Gaussians with parameter $s/\sqrt{2}$ over the cosets of $2\lat$. I.e., while the probability of their average landing in any particular coset will not in general be proportional to the Gaussian mass of the coset, the distribution \emph{inside each coset} will be exactly Gaussian. (See Lemma~\ref{lem:one_coset}.) 

This observation is sufficient to prove that 
no matter what rejection sampling procedure $f$ we use in Procedure~\ref{proc:reject}, if the input consists of independent samples from $D_{\lat -\vec{t}, s}$, the output will always be distributed as some \emph{mixture} of samples from $D_{2\lat + \vec{c} - \vec{t}, s/2^{\ell/2}}$ over the cosets $\vec{c} \in \lat/(2\lat)$. I.e., while the output distribution might distribute weight amongst the cosets differently, if we condition on a fixed number of vectors landing in each coset, the output will always be distributed as independent discrete Gaussian vectors with parameter $s/2^{\ell/2}$. It follows immediately that ``rejection sampling cannot help us.'' In particular, the probability that the output of Procedure~\ref{proc:reject} will contain a particular vector (say a shortest non-zero vector) with any rejection sampling procedure $f$ will never be greater than the probability that we would see that vector without rejection sampling (i.e., when $f$ is the trivial function that outputs $\{1,\ldots, M\}$).\footnote{%
	Notice that this property is far from obvious without the observation that the output distribution is always a mixture of Gaussians over the cosets. For example, if we modified Procedure~\ref{proc:reject} so that $f$ acted on the $\vec{X}_i$ themselves, rather than just their cosets mod $2\lat$, then this property would no longer hold.} 
See Corollary~\ref{cor:mixtures} and Theorem~\ref{thm:meta_thm} for more detail.

\subsection{An open problem---towards a \texorpdfstring{$2^{n/2}$}{2n/2}-time algorithm}

Our result shows that all known applications of the $2^{n + o(n)}$-time discrete Gaussian sampling algorithms in~\cite{ADRS15,ADS15} work just as well if we remove the rejection sampling procedure from these algorithms. This in particular includes the SVP application mentioned in Theorem~\ref{thm:SVP_intro} and the approximate CVP application mentioned in Theorem~\ref{thm:CVP_intro}. (More generally, we can remove the rejection sampling procedure from any application that simply relies on finding a set of vectors with a certain property in the output distribution, such as a shortest non-zero vector, all shortest non-zero vectors, a vector that is close to a shortest lattice vector in $\lat - \vec{t}$, etc.)

However,~\cite{ADRS15} also presents a $2^{n/2+o(n)}$-time algorithm that samples from $D_{\lat - \vec{t}, s}$ as long as the parameter $s > 0$ is not too small. (In particular, we need $s \geq \sqrt{2} \eta_{1/2}(\lat)$, where $\eta_{1/2}(\lat)$ is the \emph{smoothing parameter} of the lattice. See~\cite{ADRS15} or~\cite{NSDthesis} for the details.) This algorithm is similar to the $2^{n + o(n)}$-time algorithms in that it starts with independent discrete Gaussian vectors with some high parameter, and it gradually lowers the parameter using a rejection sampling procedure together with a procedure that takes the averages of pairs of vectors that lie in the same coset modulo some sublattice. But, it fails for smaller parameters specifically because the rejection sampling procedure that it uses must throw out too many vectors in this case. (In~\cite{NSDthesis}, we use a different rejection sampling procedure that never throws away too many vectors, but we do not know how to implement it in $2^{n/2+o(n)}$ time for small parameters $s < \sqrt{2} \eta_{1/2}(\lat)$.) If we could find a suitable variant of this algorithm that works for small parameters, we would be able to solve SVP in $2^{n/2 + o(n)}$ time. 

So, we are naturally very interested in understanding what happens when we simply remove the rejection sampling procedure from this algorithm. And, the fact that this works out so nicely for the $2^{n + o(n)}$-time algorithm works seems quite auspicious! Unfortunately, we are unable to say very much at all about the resulting distribution in the $2^{n/2 + o(n)}$-time case.\footnote{After one step of ``pairing and averaging,'' we know exactly the distribution that we get, and it \emph{is} a weighted combination of Gaussians over the cosets of a certain sublattice! This seems quite auspicious. Unfortunately, the particular sublattice is not the same sublattice that we use to pair the vectors in the next step, and we therefore are unable to say much at all about what happens even after two steps.} So, we leave the study of this distribution as an open problem.

\subsection*{Organization}

In Section~\ref{sec:prelims}, we review a few basic facts necessary to prove our main ``meta-theorem,'' Theorem~\ref{thm:meta_thm}, which shows that ``rejection sampling is unnecessary.'' In Section~\ref{sec:average_stupid}, we finish this proof. In particular, this implies Theorem~\ref{thm:SVP_intro} and~\ref{thm:CVP_intro}. For completeness, in the appendix, we prove these theorems more directly and show the resulting algorithms in full detail.

\subsection*{Acknowledgments}

We thank Oded Regev and Daniel Dadush for many helpful discussions.

\section{Preliminaries}
\label{sec:prelims}

We write $\N := \{0,1,\ldots\}$ for the natural numbers (including zero). We make little to no distinction between a random variable and its distribution. For $\vec{x} = (x_1,\ldots, x_n) \in \R^n$, we write $\|\vec{x}\| := (x_1^2 + \cdots + x_n^2)^{1/2}$ for the Euclidean norm of $\vec{x}$. For any set $S$, we write $S^* := \{(x_1,\ldots, x_M) \ : \ x_i \in S \}$ for the set lists over $S$ of finite length. (The order of elements in a list $\mathcal{M} \in S^*$ will never concern us. We could therefore instead use multisets.)

\subsection{Lattices}

A lattice $\lat \subset \R^n$ is the set of integer linear combinations
\[
\lat := \{ a_1 \vec{b}_1 + \cdots + a_n \vec{b}_n \ : \ a_i \in \Z\}
\]
of some linearly independent basis vectors $\basis := (\vec{b}_1,\ldots, \vec{b}_n)$. We sometimes write $\lat(\basis)$ for the lattice spanned by $\basis$.

We write $\lat/(2\lat)$ for the set of cosets of $\lat$ over $2\lat$. E.g., if $\vec{b}_1,\ldots, \vec{b}_n$ is a basis for $\lat$, then each coset $\vec{c} \in \lat/(2\lat)$ corresponds to a unique vector $a_1 \vec{b}_1 + \cdots + a_n \vec{b}_n$ with $a_i \in \{0,1\}$, and this correspondence is a bijection. Notice that the cosets in $\lat/(2\lat)$ have a group structure under addition that is isomorphic to $\Z_2^n$.
	
\subsection{The discrete Gaussian}

For a parameter $s > 0$ and vector $\vec{x} \in \R^n$, we write
\[
\rho_s(\vec{x}) := \exp(-\pi \|\vec{x}\|^2/s^2)
\]
for the \emph{Gaussian mass of $\vec{x}$ with parameter $s > 0$}. Up to scaling, the Gaussian mass is the unique function on $\R^n$ that is invariant under rotations and a product function. In particular, it satisfies the following nice rotation identity,
\begin{equation}
\label{eq:rotate}
\rho_{s}(\vec{x}) \rho_s(\vec{y}) = \rho_{\sqrt{2} s}(\vec{x} + \vec{y}) \rho_{\sqrt{2} s}(\vec{x} - \vec{y})
\end{equation}
for any parameter $s > 0$ and vectors $\vec{x}, \vec{y} \in \R^n$.
This identity is fundamental to the results of~\cite{ADRS15,ADS15}. (See~\cite{riemann17,NSDthesis} for a more detailed description of this connection and some additional results.)

We extend the Gaussian mass to a shift $\vec{t} \in \R^n$ of a lattice $\lat \subset \R^n$ in the natural way,
\[
\rho_s(\lat - \vec{t}) := \sum_{\vec{y} \in \lat} \rho_s(\vec{y} - \vec{t})
\; ,
\]
and we call this the \emph{Gaussian mass of $\lat - \vec{t}$ with parameter $s$}.

We will need the following identity from~\cite{ADS15}. (See~\cite{riemann17, NSDthesis} for a much more general identity.)

\begin{lemma}
	\label{lem:collision_identity}
	For any lattice $\lat \subset \R^n$, shift $\vec{t}$, and parameter $s > 0$, we have
	\[
	\sum_{\vec{c} \in \lat/(2\lat)} \rho_s(2\lat + \vec{c} - \vec{t})^2 = \rho_{s/\sqrt{2}}(\lat) \rho_{s/\sqrt{2}}(\lat - \vec{t})
	\; .
	\]
\end{lemma}
\begin{proof}
	We have
	\begin{align*}
		\sum_{\vec{c} \in \lat/(2\lat)} \rho_s(2\lat + \vec{c} - \vec{t})^2 
		&= \sum_{\vec{c} \in \lat/(2\lat)} \sum_{\vec{y}_1, \vec{y}_2 \in \lat} \rho_s(2\vec{y}_1 + \vec{c} - \vec{t})\rho_s(2\vec{y}_2 + \vec{c} - \vec{t})\\
		&= \sum_{\vec{c} \in \lat/(2\lat)} \sum_{\vec{y}_1, \vec{y}_2 \in \lat} \rho_{s/\sqrt{2}}(\vec{y}_1 + \vec{y}_2 + \vec{c} - \vec{t})\rho_{s/\sqrt{2}}(\vec{y}_1 - \vec{y}_2)\\
		&= \sum_{\vec{c} \in \lat/(2\lat)} \sum_{\vec{w}, \vec{y}_1 \in \lat} \rho_{s/\sqrt{2}}(2\vec{y}_1 - \vec{w} + \vec{c} - \vec{t})\rho_{s/\sqrt{2}}(\vec{w})\\
		&= \rho_{s/\sqrt{2}}(\lat - \vec{t}) \sum_{\vec{w} \in \lat} \rho_{s/\sqrt{2}}(\vec{w})\\
		&= \rho_{s/\sqrt{2}}(\lat - \vec{t}) \rho_{s/\sqrt{2}}(\lat)
		\; ,
	\end{align*}
	as needed.
\end{proof}

\subsection{Dominating distributions}

Intuitively, we say that some random list $\mathcal{M} \in S^*$ dominates another random list $\mathcal{M}' \in S^*$ if for every fixed list $\mathcal{S} \in S^*$, ``$\mathcal{M}$ is at least as likely to contain $\mathcal{S}$ as a subsequence as $\mathcal{M}'$ is.''

\begin{definition}[Dominating distribution]
	\label{def:dominates}
	For some finite set $S$ (which we identify with $\{1,\ldots, N\}$ without loss of generality) and two random lists $\mathcal{M} := (X_1,\ldots, X_M) \in S^*$ and $\mathcal{M}' := (X_1',\ldots, X_{M'}') \in S^{*}$ (where $M$ and $M'$ might themselves be random variables), we say that $\mathcal{M}$ \emph{dominates} $\mathcal{M}'$ if for any $(k_1,\ldots, k_N) \in \N^N$,
	\[
	\Pr[|\{ j \ : \ X_j = i \}| \geq k_i,\ \forall i] \geq \Pr[|\{ j \ : \ X_j' = i \}| \geq k_i,\ \forall i]
	\; .
	\]
\end{definition}

We note the following basic facts about dominant distributions, which show that domination yields a partial order over random variables on $S^*$, and that this partial order behaves nicely under taking sublists.

\begin{fact}
	\label{fact:monotone}
	For any finite set $S $ and random variable $\mathcal{M} \in S^*$ that dominates some other random variable $\mathcal{M}' \in S^{*}$,
	\begin{enumerate}
		\item $\mathcal{M}$ dominates itself;
		\item if $\mathcal{M}'$ dominates some random variable $\mathcal{M}'' \in S^{*}$, then $\mathcal{M}$ also dominates $\mathcal{M}''$; and
		\item for any function $f : S^* \to S^*$ that maps a list of elements to a sublist, $\mathcal{M}$ dominates $f(\mathcal{M}')$.
	\end{enumerate}
\end{fact}

\section{No need for rejection!}
\label{sec:average_stupid}

We now show our main observation: if $\vec{X}_1 \in \lat - \vec{t}$ and $\vec{X}_2 \in \lat - \vec{t}$ are sampled from the discrete Gaussian over a fixed coset $2\lat + \vec{c} - \vec{t}$ for some $\vec{c} \in \lat/(2\lat)$, then their average $(\vec{X}_1 + \vec{X}_2)/2$ is distributed as a mixture of Gaussians over the cosets $2\lat + \vec{d} - \vec{t}$ for $\vec{d} \in \lat/(2\lat)$ with parameter lowered by a factor of $\sqrt{2}$.

\begin{lemma}
	\label{lem:one_coset}
For any lattice $\lat \subset \R^n$, shift $\vec{t} \in \R^n$, parameter $s > 0$, coset $\vec{c} \in \lat/(2\lat)$, $s > 0$, and $\vec{y} \in \lat$, we have
\[
\Pr_{\vec{X}_1,\vec{X}_2 \sim D_{2\lat + \vec{c} - \vec{t},s}}[(\vec{X}_1 + \vec{X}_2)/2 = \vec{y} - \vec{t}] = \rho_{s/\sqrt{2}}(\vec{y} - \vec{t}) \cdot \frac{   \rho_{s/\sqrt{2}}(2\lat + \vec{c} + \vec{y})}{\rho_s(2\lat + \vec{c} - \vec{t})^2}
\; .
\]
In particular, for any $\vec{d} \in \lat/(2\lat)$ and $\vec{y} \in 2\lat + \vec{d}$,
\[
\Pr_{\vec{X}_1,\vec{X}_2 \sim D_{2\lat + \vec{c} - \vec{t},s}}[(\vec{X}_1 + \vec{X}_2)/2 = \vec{y} - \vec{t}\ |\ (\vec{X}_1 + \vec{X}_2)/2 \in 2\lat + \vec{d} - \vec{t}] = \frac{\rho_{s/\sqrt{2}}(\vec{y} - \vec{t})}{ \rho_{s/\sqrt{2}}(2 \lat + \vec{d}- \vec{t})}
\; .
\]
\end{lemma}
\begin{proof}
	We have
	\begin{align*}
	&\rho_s(2\lat + \vec{c} - \vec{t})^2 \cdot \Pr_{\vec{X}_1,\vec{X}_2 \sim D_{2\lat + \vec{c} - \vec{t},s}}[(\vec{X}_1 + \vec{X}_2)/2 = \vec{y} - \vec{t}] \\
	&\qquad = \sum_{\vec{x} \in 2\lat + \vec{c}}\rho_s(\vec{x} - \vec{t}) \rho_s(2\vec{y} -\vec{x} - \vec{t})\\
	&\qquad = \rho_{s/\sqrt{2}}(\vec{y} -\vec{t}) \sum_{\vec{x} \in 2\lat + \vec{c}} \rho_{s/\sqrt{2}}(\vec{x} - \vec{y}) &\text{(Eq.~\eqref{eq:rotate})}\\
	&\qquad = \rho_{s/\sqrt{2}}(\vec{y} -\vec{t}) \rho_{s/\sqrt{2}}(2\lat + \vec{c} + \vec{y})
	\; ,
	\end{align*}
	as needed. The ``in particular'' then follows from the fact that $\rho_{s/\sqrt{2}}(2\lat + \vec{c} + \vec{y}) = \rho_{s/\sqrt{2}}(2\lat + \vec{c} + \vec{d})$ is constant for $\vec{y} \in 2\lat + \vec{d}$ for some fixed $\vec{d} \in \lat/(2\lat)$.
\end{proof}

\noindent Lemma~\ref{lem:one_coset} motivates the following definition, which captures a key property of the distribution described in Lemma~\ref{lem:one_coset}.

\begin{definition}
	\label{def:mixed_gaussian}
	For a lattice $\lat \subset \R^n$, shift $\vec{t} \in \R^n$, and parameter $s > 0$  we say that the random list $(\vec{X}_1, \ldots, \vec{X}_M) \in (\lat - \vec{t})^*$ is a \emph{mixture of independent Gaussians over $\lat - \vec{t}$ with parameter $s$} if 
	the ``distributions within the cosets of $2\lat$'' are independent Gaussians with parameter $s$. I.e., for any list of cosets $(\vec{c}_1,\ldots, \vec{c}_M)\in ((\lat - \vec{t})/(2\lat))^*$ mod $2\lat$, if we \emph{condition on} $\vec{X}_i \in 2\lat + \vec{c}_i$ for all $i$, then the $\vec{X}_i$ are independent with $\vec{X}_i \sim D_{2\lat + \vec{c}_i, s}$.
	
	We call $(2\lat + \vec{X}_1,\ldots, 2\lat + \vec{X}_M)$ the \emph{coset distribution} of the $\vec{X}_i$. We say that a mixture of independent Gaussians $\mathcal{M}$ over $\lat - \vec{t}$ with parameter $s > 0$ dominates another, $\mathcal{M}'$, if the coset distribution of $\mathcal{M}$ dominates the coset distribution of $\mathcal{M}'$ (as in Definition~\ref{def:dominates}).
\end{definition}

In other words, mixtures of independent Gaussians are exactly the distributions obtained by first sampling $(\vec{c}_1,\ldots, \vec{c}_M) \in ((\lat - \vec{t})/(2\lat))^*$ from some arbitrary coset distributions and then sampling $\vec{X}_i \sim D_{\lat + \vec{c}_i, s}$ independently for each $i$.  We now list some basic facts that follow from what we have done so far.
\begin{corollary}[Properties of mixtures of Gaussians and Procedure~\ref{proc:pair_and_average}]
	\label{cor:mixtures}
	For any lattice $\lat \subset \R^n$, shift $\vec{t} \in \R^n$, and parameter $s > 0$,
	\begin{enumerate}
		\item \label{item:mixture_defined_by_cosets} a mixture of independent Gaussians over $\lat - \vec{t}$ with parameter $s$ is uniquely characterized by its coset distribution;
		\item \label{item:mixtures_preserved} if we apply Procedure~\ref{proc:pair_and_average} to a mixture of independent Gaussians over $\lat - \vec{t}$ with parameter $s$, the result will be a mixture of Gaussians over $\lat - \vec{t}$ with parameter $s/\sqrt{2}$; 
		\item \label{item:dominated_preserved} Procedure~\ref{proc:pair_and_average} preserves domination---i.e., if we apply Procedure~\ref{proc:pair_and_average} to two mixtures $\mathcal{M}, \mathcal{M}'$ of Gaussians over $\lat - \vec{t}$ with parameter $s$ and $\mathcal{M}$ dominates $\mathcal{M}'$, then the output of Procedure~\ref{proc:pair_and_average} on input $\mathcal{M}$ will dominate that of $\mathcal{M}'$; and
		\item \label{item:squared} if $\vec{X}_1, \vec{X}_2$ are a mixture of independent Gaussians over $\lat - \vec{t}$ with parameter $s$ with coset distribution given by $\vec{X}_1 \equiv \vec{X}_2 \bmod 2\lat$ and
		\[
		\Pr[2\lat + \vec{X}_1 + \vec{t} = \vec{c}] = \frac{\rho_s(2\lat + \vec{c} - \vec{t})^2}{\sum_{\vec{d} \in \lat/(2\lat)}\rho_s(2\lat + \vec{d} - \vec{t})^2}
		\; 
		\]
		for any $\vec{c} \in \lat/(2\lat)$,
		then their average $(\vec{X}_1 + \vec{X}_2)/2$ is distributed exactly as $D_{\lat - \vec{t}, s/\sqrt{2}}$.
	\end{enumerate}
\end{corollary}
\begin{proof}
	Item~\ref{item:mixture_defined_by_cosets} follows immediately from the definition of a mixture of Gaussians. Items~\ref{item:mixtures_preserved} and~\ref{item:dominated_preserved} are immediate consequences of Lemma~\ref{lem:one_coset}.
	
	For Item~\ref{item:squared}, we apply Lemma~\ref{lem:one_coset} to see that for any $\vec{y} \in \lat$,
	\begin{align*}
	\Pr[(\vec{X}_1 + \vec{X}_2)/2 = \vec{y} - \vec{t}] &= \frac{\rho_{s/\sqrt{2}}(\vec{y} - \vec{t})}{\sum_{\vec{d} \in \lat/(2\lat)}\rho_s(2\lat + \vec{d} - \vec{t})^2}\sum_{\vec{c} \in \lat/(2\lat)}    \rho_{s/\sqrt{2}}(2\lat + \vec{c} + \vec{y})\\
	&= \frac{\rho_{s/\sqrt{2}}(\vec{y} - \vec{t})}{\sum_{\vec{d} \in \lat/(2\lat)}\rho_s(2\lat + \vec{d} - \vec{t})^2} \cdot \rho_{s/\sqrt{2}}(\lat)
	\; .
	\end{align*}
	The result then follows from Lemma~\ref{lem:collision_identity}. (Indeed, summing the left-hand side and the right-hand side over all $\vec{y} \in \lat$ gives a proof of Lemma~\ref{lem:collision_identity}.)
\end{proof}

In~\cite{ADRS15,ADS15}, we performed a careful rejection sampling procedure $f$ in Procedure~\ref{proc:reject} so that, at each step of the algorithm, the output was distributed exactly as $D_{\lat - \vec{t},s/2^{i/2}}$ (up to some small statistical distance). In particular, we applied the rejection sampling procedure guaranteed by Theorem~\ref{thm:square_sampler} to obtain independent vectors distributed as in Item~\ref{item:squared}, which yield independent Gaussians with a lower parameter when combined as in Procedure~\ref{proc:pair_and_average}. But, Corollary~\ref{cor:mixtures} makes this unnecessary. Indeed, Corollary~\ref{cor:mixtures} shows that ``any collection of vectors that can be found with any rejection sampling procedure can be found without it.'' The following meta-theorem makes this formal.

\begin{theorem}
	\label{thm:meta_thm}
	For any (possibly randomized) rejection function $f$ mapping lists of cosets modulo $2\lat$ to a subset of indices (as in Procedure~\ref{proc:reject}), let $\mathcal{A}$ be the algorithm defined in Procedure~\ref{proc:reject}. Let $\mathcal{A}'$ be the same algorithm with $f$ replaced by the trivial function that just outputs all indices (i.e., $\mathcal{A}'$ just repeatedly runs Procedure~\ref{proc:pair_and_average} with no rejection).
	
	Then, for any lattice $\lat \subset \R^n$, shift vector $\vec{t} \in \R^n$, parameter $s > 0$, if $\mathcal{A}$ and $\mathcal{A}'$ are each called on input $\ell \geq 1$ and a list of $M \geq 2$ independent samples from $D_{\lat - \vec{t}, s}$, the resulting output distributions will be mixtures of independent Gaussians over $\lat - \vec{t}$ with parameter $s/2^{\ell/2}$. Furthermore, the distribution corresponding to $\mathcal{A}'$ will dominate the distribution corresponding to $\mathcal{A}$.  In particular, for any finite set $S \subset \lat - \vec{t}$,
	\[
	\Pr_{\vec{X}_1,\ldots, \vec{X}_M \sim D_{\lat - \vec{t}, s}}[S \subseteq \mathcal{A}(\ell, \vec{X}_1,\ldots, \vec{X}_M)] \leq \Pr_{\vec{X}_1,\ldots, \vec{X}_M \sim D_{\lat - \vec{t}, s}}[S \subseteq \mathcal{A}'(\ell, \vec{X}_1,\ldots, \vec{X}_M)]
	\; .
	\]
\end{theorem}
\begin{proof}
	Notice that, since $f$ only acts on the cosets of the $\vec{X}_i$, $f$ ``preserves mixtures of independent Gaussians.'' I.e., if $(\vec{X}_1,\ldots, \vec{X}_{M'})$ is some mixture of independent Gaussians over $\lat - \vec{t}$ with parameter $s' > 0$ and $(j_1,\ldots, j_m) \leftarrow f(2\lat + \vec{X}_1,\ldots, 2\lat + \vec{X}_{M'})$, then $(\vec{X}_{j_1},\ldots, \vec{X}_{j_m})$ is also a mixture of independent Gaussians over $\lat - \vec{t}$ with parameter $s'$. (Notice that this would \emph{not} be true if $f$ acted on vectors, rather than cosets.)
	Similarly, by Item~\ref{item:mixtures_preserved}, Procedure~\ref{proc:pair_and_average} maps mixtures of independent Gaussians over $\lat - \vec{t}$ with parameter $s'$ to mixtures with parameter $s'/\sqrt{2}$. It follows that for both $\mathcal{A}$ and $\mathcal{A}'$, after the $i$th step of the algorithm, the list of vectors is a mixture of Gaussians over $\lat - \vec{t}$ with parameter $s/2^{i/2}$. And, the same holds after the application of $f$ in algorithm $\mathcal{A}$. Therefore, the only question is the coset distributions. 
	
	By Fact~\ref{fact:monotone}, we see that $(\vec{X}_1,\ldots, \vec{X}_M)$ dominates $(\vec{X}_{j_1},\ldots, \vec{X}_{j_m})$. Therefore, by Item~\ref{item:dominated_preserved}, the distribution of vectors corresponding to $\mathcal{A}'$ dominates the distribution of $\mathcal{A}$ after the first step. If we assume for induction that, after the $(i - 1)$st step, the distribution of vectors corresponding to $\mathcal{A}'$ dominates the distribution corresponding to $\mathcal{A}$, then the exact same argument together with another application of Fact~\ref{fact:monotone} shows that the same holds after step $i$. The result follows.
\end{proof}

Theorem~\ref{thm:meta_thm}, together with the corresponding algorithms in~\cite{ADRS15,ADS15}, immediately implies Theorems~\ref{thm:SVP_intro} and~\ref{thm:CVP_intro}. For completeness, we give more direct proofs of these theorems in the appendix, more-or-less recreating the corresponding proofs in~\cite{ADRS15,ADS15}.

\bibliographystyle{alpha}
\newcommand{\etalchar}[1]{$^{#1}$}
\def\cprime{$'$}

\appendix

\section{Additional preliminaries}

We will need some additional preliminaries. We write 
\[
\lambda_1(\lat) := \min_{\vec{y} \in \lat \setminus \{\vec0\}} \|\vec{y}\|
\]
for the length of the shortest non-zero vector in the lattice. And, for a target vector $\vec{t} \in \R^n$, we write 
\[
\dist(\vec{t},\lat ) := \min_{\vec{y} \in \lat} \|\vec{y} - \vec{t}\|
\]
for the distance from $\vec{t}$ to the lattice. Notice that this is the same as the length of the shortest vector in $\lat - \vec{t}$.

\subsection{Some known algorithms}

We will need the famous result of
Lenstra, Lenstra, and Lov{\'a}sz~\cite{LLL82}.

\begin{theorem}[\cite{LLL82}]
	\label{thm:LLL}
	There is an efficient algorithm that take as input a lattice $\lat \subset \R^n$ and outputs $\tilde{\lambda} > 0$ with
	\[
	\lambda_1(\lat) \leq \tilde{d} \leq 2^{n/2} \lambda_1(\lat)
	\; .
	\]
\end{theorem}

We will also need the following celebrated result due to Babai~\cite{Bab86}.

\begin{theorem}[\cite{Bab86}]
	\label{thm:babai}
	There is an efficient algorithm that takes as input a lattice $\lat \subset \R^n$ and a target $\vec{t} \in \R^n$ and outputs $\tilde{d} > 0$ with 
	\[
	\dist(\vec{t}, \lat) \leq \tilde{d} \leq 2^{n/2} \dist(\vec{t}, \lat)
	\; .
	\]
\end{theorem}

\subsection{The distribution of disjoint pairs}

Recall that Procedure~\ref{proc:pair_and_average} takes the $T_i$ elements from the $i$th coset and converts them into $\floor{T_i/2}$ disjoint pairs. Therefore, for a list $\mathcal{M} := (X_1,\ldots, X_M) \in S^*$ over some finite set $S$, we write $\floor{\mathcal{M}/2}$ for the random variable obtained as in Procedure~\ref{proc:pair_and_average}. I.e., up to ordering (which does not concern us), $\floor{\mathcal{M}/2} := (X_1',\ldots, X_{M'}') \in (S \times S)^*$ is defined by 
\[
|\{ j \ : \ X_j' = (s,s)\}| = \floor{|\{ j \ : \ X_j = s\}|/2}
\]
for each $s \in S$.

	\begin{theorem}[{\cite[Theorem 3.3]{ADRS15}}]
		\label{thm:square_sampler}
		For any probabilities, $p_1,\ldots, p_N \in [0,1]$ with $\sum p_i = 1$, integer $M$, and $\kappa \geq \Omega(\log M)$ (the confidence parameter) with $M \geq 10\kappa^2/\pmax$, let $\mathcal{M} = (X_1,\ldots, X_M) \in \{1,\ldots, N\}^M$ be the distribution obtained by sampling each $X_j$ independently from the distribution that assigns to element $i$ probability $p_i$.  Then, there exists a rejection sampling procedure that, up to statistical distance $\exp(-\Omega(\kappa))$, maps $\mathcal{M}$ to the distribution $\mathcal{M}' := (X_1',\ldots, X_M') \in \{(1,1),\ldots, (N,N) \}^{M'}$ obtained by sampling each pair $X_j$ independently from the distribution that assigns to the pair $(i,i)$ probability $p_i^2/\pcol$, where
		\[
		M' := \left\lceil M \cdot \frac{\pcol}{32 \kappa \pmax} \right\rceil
		\; ,
		\] 
		$\pmax := \max p_i$, and $\pcol := \sum p_i^2$.
	\end{theorem}

	\begin{corollary}
		\label{cor:square_sampler}
		For any probabilities, $p_1,\ldots, p_N \in [0,1]$ with $\sum p_i = 1$, integer $M$, and $\kappa \geq \Omega(\log M)$ (the confidence parameter) with $M \geq 10\kappa^2/\pmax$, let $\mathcal{M} := (X_1,\ldots, X_M) \in \{1,\ldots, N\}^M$ be the distribution obtained by sampling each $X_j$ independently from the distribution that assigns to element $i$ probability $p_i$, and let $\mathcal{M}' := (X_1',\ldots, X_{M'}')\in \{(1,1),\ldots, (N,N) \}^{M'}$ be the distribution obtained by sampling each pair $X_j'$ independently from the distribution that assigns to the pair $(i,i)$ probability $p_i^2/\pcol$, where 
		\[
		M' := \left\lceil M \cdot \frac{\pcol}{32 \kappa \pmax} \right\rceil
		\; ,
		\] 
		$\pmax := \max p_i$, and $\pcol := \sum p_i^2$.
		Then, $\mathcal{M}$ dominates $\mathcal{M}'$.
	\end{corollary}

\subsection{Additional facts about the discrete Gaussian} 

We will also need some additional facts about the discrete Gaussian.

\begin{lemma}[\cite{banaszczyk}]
	\label{lem:banaszczyk_growth}
	For any lattice $\lat \subset \R^n$, parameter $s \geq 1$, and shift $\vec{t} \in \R^n$, $\rho_s(\lat - \vec{t}) \leq s^n \rho(\lat)$.
\end{lemma}

The following theorem shows that, if the parameter $s$ is appropriately small, then $D_{\lat - \vec{t}, s} + \vec{t}$ will be an approximate closest vector to $\vec{t}$, with approximation factor roughly $1+\sqrt{n} s/\dist(\vec{t}, \lat)$. (This is a basic consequence of Banaszczyk's celebrated theorem~\cite{banaszczyk}.)

\begin{proposition}[{\cite[Corollary 1.3.11]{NSDthesis}, see also \cite[Corollary 2.8]{ADS15}}]
	\label{prop:DGS_approx_CVP}
	For any lattice $\lat \subset \R^n$, parameter $s > 0$, shift $\vec{t} \in \R^n$, and radius $r > \sqrt{n/(2\pi)} \cdot s$, with $r > \dist(\vec{t}, \lat)$ and
	\[
	r^2 > \dist(\vec{t}, \lat)^2 + \frac{n s^2}{\pi} \cdot \log (2\pi \dist(\vec{t}, \lat)^2/(ns^2))
	\; ,
	\]
	we have
	\[
	\Pr_{\vec{X} \sim D_{\lat - \vec{t}, s}}[\|\vec{X}\| > r] <  (2e)^{n/2+1}\exp(-\pi y^2/2)
	\; ,
	\]
	where $y := \sqrt{r^2 - \dist(\vec{t}, \lat)^2}/s$.
\end{proposition}

The next theorem shows that exponentially many samples from $D_{\lat, s}$ with $s \approx \lambda_1(\lat)/\sqrt{n}$ is sufficient to find a shortest non-zero lattice vector.

\begin{proposition}[{\cite[Proposition 4.3]{ADRS15}}]
	\label{prop:SVPtoDGS}
	For any lattice $\lat \subset \R^n$, and parameter 
	\[
	s := \sqrt{2^{0.198}\pi e/n} \cdot \lambda_1(\lat)
	\;, 
	\]
	we have
	\[
	\Pr_{\vec{X} \sim D_{\lat, s}}[\|\vec{X}\| = \lambda_1(\lat)] \geq 1.38^{-n - o(n)}
	\; .
	\]
\end{proposition}

The next corollary follows immediately from Proposition~\ref{prop:SVPtoDGS} and Lemma~\ref{lem:banaszczyk_growth}.

\begin{corollary}
	\label{cor:approx_right_parameter}
	For any lattice $\lat \subset \R^n$, and parameter 
	\[
	\sqrt{2^{0.198}\pi e/n} \cdot \lambda_1(\lat) \leq s \leq  1.01 \cdot \sqrt{2^{0.198}\pi e/n} \cdot \lambda_1(\lat)
	\;, 
	\]
	we have
	\[
	\Pr_{\vec{X} \sim D_{\lat, s}}[\|\vec{X}\| = \lambda_1(\lat)] \geq 1.4^{-n - o(n)}
	\; .
	\]
\end{corollary}

We will also need the following result from~\cite{ADS15}, which is an immediate consequence of the main identity in~\cite{riemann17}. (See also~\cite{NSDthesis}.)

\begin{lemma}[{\cite[Corollary 3.3]{ADS15}}]
	\label{lem:rotation_cor_shifted}
	For any lattice $\lat \subset \R^n$, shift $\vec{t} \in \R^n$, and parameter $s > 0$, we have
	\[
	\max_{\vec{c} \in \lat/(2\lat)} \rho_s(2\lat + \vec{c} - \vec{t})^2 \leq \rho_{s/\sqrt{2}}(\lat)\max_{\vec{c} \in \lat} \rho_{s/\sqrt{2}}(2\lat + \vec{c} - \vec{t})
	\; .
	\]
\end{lemma}

From this, we derive the following rather technical-looking inequality, which is implicit in~\cite{ADS15}. (This inequality comes up naturally in the proof of Corollary~\ref{cor:pipeline}. We separate it out here to make that proof cleaner.)

\begin{corollary}
	\label{cor:loss_factor_product}
	For any lattice $\lat \subset \R^n$, shift $\vec{t} \in \R^n$, parameter $s > 0$, and integer $\ell \geq 0$, we have
	\begin{align*}
	&\prod_{i=0}^{\ell-1} 
	\frac{ \rho_{s/2^{(i+1)/2}}(\lat - \vec{t}) \rho_{s/2^{(i+1)/2}}(\lat)}
	{\rho_{s/2^{i/2}}(\lat - \vec{t}) \cdot \max_{\vec{c} \in \lat/(2\lat)} \rho_{s/2^{i/2}}(2\lat + \vec{c} - \vec{t})}\\
	&\qquad\geq \frac{\rho_{s/2^{\ell/2}}(\lat - \vec{t})}{\max_{\vec{c} \in \lat/(2\lat)} \rho_{s/2^{\ell/2}}(2\lat + \vec{c} - \vec{t})} \cdot \frac{\max_{\vec{c} \in \lat/(2\lat)} \rho_{s}(2\lat + \vec{c} - \vec{t})}{\rho_{s}(\lat - \vec{t})}
	\; .
	\end{align*}
\end{corollary}
\begin{proof}
	From Lemma~\ref{lem:rotation_cor_shifted}, we see that for all $i$,
	\[
	\frac{\rho_{s/2^{(i+1)/2}} (\lat)}{\max_{\vec{c} \in \lat/(2\lat)} \rho_{s/2^{i/2}}(2\lat + \vec{c} - \vec{t})} \ge \frac{\max_{\vec{c} \in \lat/(2\lat)} \rho_{s/2^{i/2}}(2\lat + \vec{c} - \vec{t})}{\max_{\vec{c} \in \lat/(2\lat)} \rho_{s/2^{(i+1)/2}}(2\lat + \vec{c} - \vec{t})}
	\;.\]
	Therefore, the product in the statement of the corollary is at least
	\begin{align*}
	&\prod_{i=0}^{\ell-1} 
	\frac{ \rho_{s/2^{(i+1)/2}}(\lat - \vec{t}) \cdot \max_{\vec{c} \in \lat/(2\lat)} \rho_{s/2^{i/2}}(2\lat + \vec{c} - \vec{t})}
	{\rho_{s/2^{i/2}}(\lat - \vec{t}) \cdot \max_{\vec{c} \in \lat/(2\lat)} \rho_{s/2^{(i+1)/2}}(2\lat + \vec{c} - \vec{t})}\\
	& \qquad = \frac{\rho_{s/2^{\ell/2}}(\lat - \vec{t})}{\max_{\vec{c} \in \lat/(2\lat)} \rho_{s/2^{\ell/2}}(2\lat + \vec{c} - \vec{t})} \cdot \frac{\max_{\vec{c} \in \lat/(2\lat)} \rho_{s}(2\lat + \vec{c} - \vec{t})}{\rho_{s}(\lat - \vec{t})}
	\; ,
	\end{align*}
	where we have used the fact that this is a telescoping product.
\end{proof}

\section{Running Procedure~\ref{proc:pair_and_average} on Gaussian input}

\begin{theorem}
	\label{thm:combiner}
	For any lattice $\lat \subset \R^n$, shift $\vec{t} \in \R^n$, parameter $s > 0$, integer $M$, and confidence parameter $\kappa \geq \Omega(\log M)$,
	if $\vec{X}_1,\ldots, \vec{X}_M$ are sampled independently from $D_{\lat- \vec{t}, s}$ with
	\[
	M \geq 10\kappa^2 \cdot \frac{\rho_s(\lat - \vec{t})}{\max_{\vec{c} \in \lat/(2\lat)} \rho_s(2\lat + \vec{c} - \vec{t})}
	\; ,
	\] 
	 then the output of Procedure~\ref{proc:pair_and_average} applied to the $\vec{X}_i$ will be a mixture of independent Gaussians with parameter $s/\sqrt{2}$ that dominates the distribution of 
	\[
	M' := \left\lceil \frac{M}{32\kappa} \cdot \frac{\rho_{s/\sqrt{2}}(\lat - \vec{t}) \cdot  \rho_{s/\sqrt{2}}(\lat)}{\rho_s(\lat - \vec{t}) \cdot \max_{\vec{d} \in \lat/(2\lat)} \rho_s(2\lat + \vec{d} - \vec{t})} \right\rceil
	\]
	independent samples from $D_{\lat - \vec{t},s/\sqrt{2}}$,
	up to statistical distance
	$\exp(-\Omega(\kappa))$.
\end{theorem}
\begin{proof}
	By Item~\ref{item:mixtures_preserved} of Corollary~\ref{cor:mixtures}, the resulting distribution will in fact be a mixture of independent Gaussians over $\lat - \vec{t}$ with parameter $s/\sqrt{2}$. Notice that, if $\mathcal{M}$ is the coset distribution of $(\vec{X}_1,\ldots, \vec{X}_M)$, then Procedure~\ref{proc:pair_and_average} first maps the $\vec{X}_i$ into the mixture of independent Gaussians over $\lat - \vec{t}$ with parameter $s$ and coset distribution $\floor{\mathcal{M}/2}$ and then takes the averages of the corresponding pairs of these vectors. 
	
	We wish to apply Corollary~\ref{cor:square_sampler} over the coset distribution, with the probabilities $p_i := p_{2\lat  + \vec{c}}$ taken to be the weights of the cosets in the original distribution discrete Gaussian,
	\[
	p_{2\lat + \vec{c}} := \frac{\rho_s(2\lat + \vec{c} - \vec{t})}{\rho_s(\lat - \vec{t})}
	\;. 
	\]
	Notice that, by Lemma~\ref{lem:collision_identity},
	\[
	M' =  \left\lceil M \cdot \frac{\pcol}{32 \kappa \pmax} \right\rceil
	\;, 
	\]
	which is exactly what is needed to apply Corollary~\ref{cor:square_sampler}.
	By the corollary, up to statistical distance $\exp(-\Omega(\kappa))$ this distribution dominates the mixture of independent Gaussians over $\lat - \vec{t}$ with parameter $s$ whose coset distribution is given by $\vec{c}_{2k-1} = \vec{c}_{2k}$ for $1 \leq k \leq M'$, with the odd-indexed cosets $\vec{c}_{2k-1}$ sampled independently from the distribution that assigns to coset $\vec{c} \in \lat/(2\lat)$ probability
	\[
	\frac{p_i}{\pcol} = \frac{\rho_s(2\lat + \vec{c} - \vec{t})^2}{\sum_{\vec{d} \in \lat/(2\lat)}\rho_s(2\lat + \vec{d} - \vec{t})^2}
	\; .
	\]
	Notice that this ``squared'' distribution'' (so-called because the cosets are given weight proportional to their square) is simply $M'$ independent copies of the distribution from Item~\ref{item:squared} of Corollary~\ref{cor:mixtures}. So, if we run Procedure~\ref{proc:pair_and_average} on this ``squared'' distribution, the output will be exactly $M'$ independent samples from $D_{\lat - \vec{t},s/\sqrt{2}}$.
	
	Finally, by Fact~\ref{fact:monotone}, we see that, since the actual pairs dominate these ``squared'' pairs (up to statistical distance $\exp(-\Omega(\kappa))$), the output must dominate $M'$ independent samples from $D_{\lat - \vec{t}, s/\sqrt{2}}$.
\end{proof}
	
\begin{corollary}
	\label{cor:pipeline}
	For any lattice $\lat \subset \R^n$, shift $\vec{t} \in \R^n$, parameter $s > 0$, integer $M \geq 2$, and confidence parameter $\kappa \geq \Omega(\log M)$,
	if $\vec{X}_1,\ldots, \vec{X}_M$ are sampled independently from $D_{\lat- \vec{t}, s}$ with
	\[
	M \geq (10\kappa)^{2\ell} \cdot \frac{\rho_s(\lat - \vec{t})}{\max_{\vec{c} \in \lat/(2\lat)} \rho_s(2\lat + \vec{c} - \vec{t})}
	\; ,
	\] 
	and we apply Procedure~\ref{proc:pair_and_average} repeatedly to the $\vec{X}_i$ a total of $\ell \geq 1$ times, the result will be a mixture of independent Gaussians with parameter $s/2^{\ell/2}$ that dominates the distribution of 
	\[
	M' := \left\lceil  \frac{M}{(32\kappa)^\ell} \cdot \prod_{i=0}^{\ell-1} \frac{\rho_{s/2^{(i+1)/2}}(\lat - \vec{t}) \rho_{s/2^{(i+1)/2}}(\lat)}
	{\rho_{s/2^{i/2}}(\lat - \vec{t}) \cdot \max_{\vec{c} \in \lat/(2\lat)} \rho_{s/2^{i/2}}(2\lat + \vec{c} - \vec{t})} \right\rceil
	\]
	independent samples from $D_{\lat - \vec{t},s/2^{\ell/2}}$,
	up to statistical distance
	$\ell \exp(-\Omega(\kappa))$.
\end{corollary}
\begin{proof}
	By Item~\ref{item:mixtures_preserved} of Corollary~\ref{cor:mixtures}, the output will in fact be a mixture of independent Gaussians over $\lat - \vec{t}$ with parameter $s/2^{\ell/2}$. The only question is what the coset distribution is.
	
	To show that the coset distribution is as claimed, the idea is to simply apply Theorem~\ref{thm:combiner} $\ell$ times. In particular, we prove the result via induction on $\ell$. When $\ell = 1$, this is exactly Theorem~\ref{thm:combiner}. For $\ell > 1$, we assume the statement is true for $\ell-1$. In particular, before applying Procedure~\ref{proc:pair_and_average} the $\ell$th time, we have a mixture of independent Gaussians with parameter $s/2^{\ell/2}$ that dominates 
	\begin{align*}
	\widehat{M} &:= \left\lceil  \frac{M}{(32\kappa)^\ell} \cdot 
	\prod_{i=0}^{\ell-2} \frac{\rho_{s/2^{(i+1)/2}}(\lat - \vec{t}) \rho_{s/2^{(i+1)/2}}(\lat)}
	{\rho_{s/2^{i/2}}(\lat - \vec{t}) \cdot \max_{\vec{c} \in \lat/(2\lat)} \rho_{s/2^{i/2}}(2\lat + \vec{c} - \vec{t})} \right\rceil\\
	&\geq  10\kappa^2 \cdot \frac{\rho_s(\lat - \vec{t})}{\max_{\vec{c} \in \lat/(2\lat)} \rho_s(\lat + \vec{c} - \vec{t})} \cdot \prod_{i=0}^{\ell-2} \frac{\rho_{s/2^{(i+1)/2}}(\lat - \vec{t}) \rho_{s/2^{(i+1)/2}}(\lat)}
	{\rho_{s/2^{i/2}}(\lat - \vec{t}) \cdot \max_{\vec{c} \in \lat/(2\lat)} \rho_{s/2^{i/2}}(2\lat + \vec{c} - \vec{t})} 
	\end{align*}
	independent Gaussians up to statistical distance $(\ell -1) \exp(-\Omega(\kappa))$.  
	
	By Fact~\ref{fact:monotone}, it suffices to prove that the output of Procedure~\ref{proc:pair_and_average} on these $\widehat{M}$ samples dominates $M'$ independent samples from $D_{\lat - \vec{t},s/2^{\ell/2}}$ up to statistical distance $\exp(-\Omega(\kappa))$. Indeed, this is exactly what Theorem~\ref{thm:combiner} says, provided that 
	\[
	\widehat{M} \geq 10\kappa^2 \cdot \frac{\rho_{s/2^{(\ell-1)/2}}(\lat - \vec{t})}{\max_{\vec{c} \in \lat/(2\lat)} \rho_{s/2^{(\ell-1)/2}}(2\lat + \vec{c} - \vec{t})}
	\; .
	\]
	And, this inequality follows immediately from Corollary~\ref{cor:loss_factor_product} together with the assumed lower bound on $\widehat{M}$.
\end{proof}
	
\section{The initial distribution}

The following theorem was proven by Ajtai, Kumar, and Sivakumar~\cite{AKS01}, building on work of Schnorr~\cite{Schnorr87}.
\begin{theorem}[\cite{Schnorr87,AKS01}]
	\label{thm:BKZ}
	There is an algorithm that takes as input a lattice $\lat \subset \R^n$ and $u \geq 2$ and outputs an $u^{n/y}$-reduced basis of $\lat$ in time $\exp(O(u)) \cdot \poly(n)$, where we say that a basis $\basis = (\vec{b}_1,\ldots, \vec{b}_n)$ of a lattice $\lat$ is $\gamma$-reduced for some $\gamma \geq 1$ if 
	\begin{enumerate}
		\item $\length{\vec{b}_1} \leq \gamma \cdot \lambda_1(\lat)$; and
		\item $\pi_{\{ \vec{b_1} \}^\perp}(\vec{b}_2), \ldots, \pi_{\{ \vec{b}_1 \}^\perp}(\vec{b}_n)$ is a $\gamma$-reduced basis of $\pi_{\{ \vec{b_1} \}^\perp}(\lat)$.
	\end{enumerate}
\end{theorem}

This next theorem is originally due to~\cite{GPV08}, based on analysis of an algorithm originally studied by Klein~\cite{Klein00}. We present a slightly stronger version due to~\cite{BLPRS13} for convenience.

\begin{theorem}[{\cite[Lemma 2.3]{BLPRS13}}]
	\label{thm:GPV}
	There is a probabilistic polynomial-time algorithm that takes as input a basis $\basis $ for a lattice $\lat \subset \R^n$ with $n \geq 2$, a shift $\vec{t} \in \R^n$, and $\hat{s} > C\sqrt{ \log n} \cdot \length{\gs{\basis}}$ and outputs a vector that is distributed exactly as $D_{\lat - \vec{t}, \hat{s}}$, where $\length{\gs{\basis}} := \max \length{\gs{\vec{b}}_i}$.
\end{theorem}

 \begin{proposition}[{\cite[Proposition 4.5]{ADS15}}]
 	\label{prop:shiftedsublattice}
 	There is an algorithm that takes as input a lattice $\lat \subset \R^n $, shift $\vec{t} \in \R^n$, $r > 0$,  and parameter $u \geq 2$, such that if \[
 	r \geq u^{n/u} (1+\sqrt{n} u^{n/u}) \cdot \dist(\vec{t}, \lat)
 	\; ,
 	\] 
 	then the output of the algorithm is $\vec{y} \in \lat$ and a basis $\basis'$ of a (possibly trivial) sublattice $\lat' \subseteq \lat $ such that all vectors from $\lat - \vec{t}$ of length at most $r/u^{n/u} - \dist(\vec{t}, \lat)$ are also contained in $\lat' - \vec{y} - \vec{t}$, and $\length{\gs{\basis}'} \leq r$. The algorithm runs in time $\poly(n) \cdot 2^{O(u)}$.
 \end{proposition}
 \begin{proof}
 	On input a lattice $\lat \subset \R^n$, $\vec{t} \in \R^n$, and $r > 0$, the algorithm behaves as follows. First, it calls the procedure from  Theorem~\ref{thm:BKZ} to compute a $u^{n/u}$-HKZ basis $\basis  = (\vec{b}_1, \ldots, \vec{b}_n)$ of $\lat$. Let $(\gs{\vec{b}}_1, \ldots, \gs{\vec{b}}_n)$ be the corresponding Gram-Schmidt vectors. Let $k \geq 0$ be maximal such that $\length{\gs{\vec{b}}_i} \le r$ for $1 \le i \le k$, and let $\basis' = (\vec{b}_1, \ldots, \vec{b}_k)$.  Let $\pi_k = \pi_{\{ \vec{b}_1, \ldots, \vec{b}_k \}^\perp}$ and $\M = \pi_k(\lat)$. The algorithm then calls the procedure from Theorem~\ref{thm:BKZ} again with the same $s$ and input $\pi_k(\vec{t})$ and $\M$, receiving as output $\vec{x} = \sum_{i=k+1}^n a_i \pi_k(\vec{b}_i)$ where $a_i \in \Z$, a $\sqrt{n} u^{n/u}$-approximate closest vector to $\pi_k(\vec{t})$ in $\M$. Finally, the algorithm returns $\vec{y} = -\sum_{i=k+1}^n a_i \vec{b}_i$ and  $\basis' = (\vec{b}_1, \ldots, \vec{b}_k)$.
 	
 	The running time is clear, as is the fact that $\length{\gs{\basis'}} \leq r$. It remains to prove that $\lat' - \vec{y} - \vec{t}$ contains all sufficiently short vectors in $\lat - \vec{t}$. If $k = n$, then $\lat' = \lat$ and $\vec{y}$ is irrelevant, so we may assume that $k < n$. Note that, since $\basis$ is a $u^{n/u}$-HKZ basis, $\lambda_1(\M) \geq \length{\gs{\vec{b}}_{k+1}}/u^{n/u} > r/u^{n/u}$.  In particular, $\lambda_1(\M)  > (1+\sqrt{n} \cdot u^{n/u})\cdot \dist(\vec{t}, \lat) \geq (1+\sqrt{n} \cdot u^{n/u})\cdot \dist(\pi_k(\vec{t}), \M)$. So, there is a unique closest vector to $\pi_k(\vec{t})$ in $\M$, and by triangle inequality, the next closest vector is at distance greater than $\sqrt{n} \cdot u^{n/u}\dist(\pi_k(\vec{t}), \M)$. Therefore, the call to the subprocedure from Theorem~\ref{thm:BKZ} will output the exact closest vector $\vec{x} \in \M$ to $\pi_k(\vec{t})$.
 	
 	Let $\vec{w} \in \lat \setminus (\lat' - \vec{y})$ so that $\pi_k(\vec{w}) \neq \pi_k(-\vec{y}) = \vec{x}$. We need to show that $\vec{w} - \vec{t}$ is relatively long. Since $\basis$ is a $s^{n/s}$-HKZ basis, it follows that
 	\[
 	\length{\pi_k(\vec{w}) - \vec{x}} \geq \lambda_1(\M) > r/u^{n/u}
 	\; .
 	\]
 	Applying triangle inequality, we have
 	\begin{align*}
 	\length{\vec{w} - \vec{t}} \geq \length{\pi_k(\vec{w}) - \pi_k(\vec{t})}
 	\geq \length{\pi_k(\vec{w}) - \vec{x}} - \length{ \vec{x} - \pi_k(\vec{t})}
 	> r/u^{n/u}- \dist(\vec{t}, \lat)
 	\; ,
 	\end{align*}
 	as needed.
 \end{proof}

 \begin{corollary}[{\cite[Corollary 4.6]{ADS15}}]
 	\label{cor:start}
 	There is an algorithm that takes as input a lattice $\lat \subset \R^n$ with $n \geq 2$, shift $\vec{t} \in \R^n$, $M \in \N$ (the desired number of output vectors), and parameters $u \geq 2$ and 
 	$\hat{s} > 0$
 	and outputs $\vec{y} \in \lat$, a (possibly trivial) sublattice $\lat' \subseteq \lat$, and $M$ vectors from $\lat' - \vec{y} - \vec{t}$ such that
 	if 
 	\[
 	\hat{s} \geq 10\sqrt{n \log n} \cdot u^{2n/u} \cdot \dist(\vec{t}, \lat)
 	\; ,
 	\]
 	then the output vectors are distributed as $M$ independent samples from $D_{\lat' - \vec{y} - \vec{t}, \hat{s}}$, and $\lat' - \vec{y} - \vec{t}$ contains all vectors in $\lat - \vec{t}$ of length at most $\hat{s}/(10u^{n/u}\sqrt{\log n})$. The algorithm runs in time $\poly(n) \cdot 2^{O(u)} + \poly(n) \cdot M $. (And, if $\vec{t} = \vec0$, then $\vec{y} = \vec0$.)
 \end{corollary}
 \begin{proof}
 	The algorithm first calls the procedure from Proposition~\ref{prop:shiftedsublattice} with input $\lat$, $\vec{t}$, and 
 	\[
 	r := \frac{10\hat{s}}{\sqrt{\log n}} \geq u^{n/u} (1+\sqrt{n} u^{n/u}) \cdot \dist(\vec{t}, \lat)
 	\; ,
 	\] receiving as output $\vec{y} \in \lat$ and a basis $\basis'$ of a sublattice $\lat' \subset \lat $. It then runs the algorithm from Theorem~\ref{thm:GPV} $M$ times with input $\lat'$, $\vec{y} + \vec{t}$, and $\hat{s}$ and outputs the resulting vectors, $\vec{y}$, and $\lat'$.
 	
 	The running time is clear. By Proposition~\ref{prop:shiftedsublattice}, $\lat' - \vec{y} - \vec{t}$ contains all vectors of length at most $r/u^{n/u} - \dist(\vec{t}, \lat) \geq  \hat{s}/(10u^{n/u}\sqrt{\log n})$ in $\lat - \vec{t}$, and $\length{\gs{\basis}'} \leq r \leq C \hat{s}/\sqrt{ \log n} $. So, it follows from Theorem~\ref{thm:GPV} that the output has the correct distribution.
 \end{proof}

\section{Finishing the proof}

\begin{algorithm}
		\SetKwInOut{Input}{Input}
		\SetKwInOut{Output}{Output}
		\underline{SVP} $(\lat)$\\
		\Input{\ \ A lattice $\lat \subset \R^n$}
		\Output{\ \ A vector $\vec{y} \in \lat$ with $\|\vec{y}\| = \lambda_1(\lat)$}
		Use the procedure from Thereom~\ref{thm:LLL} to compute $\widehat{\lambda}$ with $\lambda_1(\lat) \leq \widehat{\lambda} \leq 2^{n/2} \lambda_1(\lat)$.\\
		\For{$i= 1,\ldots, 200n$}{
		Set $\lat' \subseteq \lat$ and $\vec{X}_1,\ldots, \vec{X}_M \in \lat$ to be the output of Corollary~\ref{cor:start} on input $\lat$, $\vec{t} := \vec0$, $u$, and $s_i := 1.01^{-i} \cdot \widehat{\lambda}$.\\
		\For{$j = 1,\ldots, \ell$}
		{
			$(\vec{X}_1,\ldots, \vec{X}_{M'}) \leftarrow \text{Pair\_and\_Average}(\vec{X}_{1},\ldots, \vec{X}_{M})$\\
			$M \leftarrow M'$
		}
		$\vec{Y}_i \leftarrow \argmin_{\vec{X}_j \neq \vec0} \|\vec{X}_j\|$.
		}
		Output $\argmin \|\vec{Y}_i\|$.
	\caption{\label{proc:SVP} The final $2^{n + o(n)}$-time SVP algorithm. Here $M = 2^{n + \Theta(\log^2 n)}$, $u = \Theta(n)$, and $\ell = \Theta(\log n)$.}
\end{algorithm}

\begin{theorem}[SVP algorithm]
	For any lattice $\lat \subset \R^n$, the output of Procedure~\ref{proc:SVP} on input $\lat$ will be a shortest non-zero vector in $\lat$ except with probability at most $\exp(-\Omega(n))$.
\end{theorem}
\begin{proof}
	The running time is clear. Let $\kappa = \Theta(n)$.
	Let $i$ such that $s_i/2^{\ell/2}$ satisfies the inequality in Corollary~\ref{cor:approx_right_parameter}. By Corollary~\ref{cor:start}, the $(\vec{X}_1,\ldots, \vec{X}_M)$ corresponding to this $i$ will be distributed exactly as $D_{\lat', s_i}$ where $\lat' \subseteq \lat$ contains all vectors of length at most $\lambda_1(\lat)$. So, $\lambda_1(\lat') = \lambda_1(\lat)$, and it suffices to argue that we will find a shortest vector in $\lat'$. By Corollary~\ref{cor:pipeline}, the output distribution $(X_1,\ldots, X_M)$ will be a mixture of independent Gaussians over $\lat'$ with parameter $s_i/2^{\ell/2}$ that dominates the distribution of 
	\[
	M' =  \left\lceil  \frac{M}{(32\kappa)^{\ell}} \cdot \prod_{j=0}^{\ell-1} \frac{\rho_{s_i/2^{(j+1)/2}}(\lat')^2}
	{\rho_{s_i/2^{j/2}}(\lat') \cdot \rho_{s_i/2^{(j+2)/2}}(\lat')} \right\rceil
	\]
	independent samples from $D_{\lat', s_i/2^{\ell/2}}$ up to statistical distance $\exp(-\Omega(\kappa))$,
	where we have applied Lemma~\ref{lem:banaszczyk_growth} to show that the coset with maximal mass is the central coset.
	Noting that this product is telescoping, we have
	\[
	M' =  \left\lceil  \frac{M}{(32\kappa)^\ell} \cdot \frac{\rho_{s_i/\sqrt{2}}(\lat')}{\rho_{s_i}(\lat')} \cdot \frac{\rho_{s_i/2^{\ell/2}}(\lat')}{ \rho_{s_i/2^{(\ell + 1)/2}}(\lat')} \right\rceil \geq 2^{n/2}
	\; ,
	\]
	where we have applied Lemma~\ref{lem:banaszczyk_growth}. The result then follows from Corollary~\ref{cor:approx_right_parameter}, together with the fact that $\sqrt{2} > 1.4$.
\end{proof}

\begin{algorithm}
	\SetKwInOut{Input}{Input}
	\SetKwInOut{Output}{Output}
	\underline{CVP} $(\lat, \vec{t})$\\
	\Input{\ \ A lattice $\lat \subset \R^n$ and target $\vec{t} \in \R^n$}
	\Output{\ \ A vector $\vec{y} \in \lat$ with $\|\vec{y} - \vec{t}\| \leq (1+2^{-n/\log^2n})\cdot\dist(\vec{t}, \lat)$}
	Use the procedure from Thereom~\ref{thm:babai} to compute $\hat{d}$ with $\dist(\lat, \vec{t}) \leq \hat{d} \leq 2^{n/2} \dist(\vec{t}, \lat)$.\\
	\For{$i= 1,\ldots, n$}{
		Set $\lat' \subseteq \lat$, $\vec{y} \in \lat$, and $X_1,\ldots, X_M \in \lat' - \vec{y} - \vec{t}$ to be the output of Corollary~\ref{cor:start} on input $\lat$, $\vec{t}$, $u$, and $s_i := 20n^2 \cdot 2^{-i} \cdot \hat{d}$.\\
		\For{$j = 1,\ldots, \ell$}
		{
			$(\vec{X}_1,\ldots, \vec{X}_{M'}) \leftarrow \text{Pair\_and\_Average}(\vec{X}_{1},\ldots, \vec{X}_{M})$\\
			$M \leftarrow M'$
		}
		$\vec{Y}_i \leftarrow \argmin_{\vec{X}_j} \|\vec{X}_j\|$.
	}
	Output $\vec{t} + \argmin \|\vec{Y}_i\|$.
	\caption{\label{proc:CVP} The final $2^{n + o(n)}$-time SVP algorithm. Here $M = 2^{n + \Theta(n/\log n)}$, $u = \Theta(n)$, and $\ell = \Theta(n/\log^2 n)$.}
\end{algorithm}

\begin{theorem}[CVP algorithm]
	For any lattice $\lat \subset \R^n$ and $\vec{t} \in\R^n$, the output of Procedure~\ref{proc:CVP} on input $\lat$ and $\vec{t}$ will a vector $\vec{y} \in \lat$ with $\|\vec{y} - \vec{t}\| \leq (1+\exp(-\Omega(n/\log^2n))) \cdot \dist(\vec{t}, \lat)$, except with probability at most $\exp(-\Omega(n))$.\footnote{It is immediate from the proof that this result can be extended to work for any approximation factor $\gamma$ with $\gamma > 1 + \exp(-o(n/\log n))$, by taking $\ell = o(n/\log n)$ and $M = 2^{n + o(n)}$ to be slightly larger.}
\end{theorem}
\begin{proof}
	The running time is clear. Let $\kappa = \Theta(n)$.
	Let $i$ such that 
	\[
	10\sqrt{n \log n} \cdot u^{2n/u} \cdot \dist(\vec{t}, \lat) \leq s_i \leq 20\sqrt{n \log n} \cdot u^{2n/u} \cdot \dist(\vec{t}, \lat)
	\; .
	\] 
	By Corollary~\ref{cor:start}, the $(\vec{X}_1,\ldots, \vec{X}_M)$ corresponding to this $i$ will be distributed exactly as $D_{\lat' - \vec{y} - \vec{t}, s_i}$ where $\lat' - \vec{y} - \vec{t} \subseteq \lat - \vec{t}$ contains all vectors of length at most $\dist(\vec{t}, \lat)$. So, it suffices to argue that we will find a $(1+2^{-n/\log^2 n})$-approximate shortest vector in $\lat' - \vec{y} - \vec{t}$. By Corollary~\ref{cor:pipeline}, the output distribution $(X_1,\ldots, X_M)$ will be a mixture of independent Gaussians over $\lat' - \vec{y} - \vec{t}$ with parameter $s_i/2^{\ell/2}$ that dominates the distribution of 
	\[
	M' =  \left\lceil  \frac{M}{(32\kappa)^\ell} \cdot \prod_{j=0}^{\ell-1} \frac{\rho_{s_i/2^{(j+1)/2}}(\lat' - \vec{y} - \vec{t}) \rho_{s_i/2^{(j+1)/2}}(\lat)}
	{\rho_{s_i/2^{j/2}}(\lat' - \vec{y} - \vec{t}) \cdot \max_{\vec{c} \in \lat'/(2\lat')} \rho_{s_i/2^{j/2}}(2\lat' + \vec{c} - \vec{y} - \vec{t})} \right\rceil \geq 1
	\]
	independent samples from $D_{\lat' - \vec{y} - \vec{t}, s_i/2^{\ell/2}}$ up to statistical distance $\exp(-\Omega(\kappa))$,
	where we have applied Corollary~\ref{cor:loss_factor_product}.
	
	Notice that $s_i/2^{\ell/2} < \exp(-\Omega(n/\log^2 n)) \dist(\vec{t}, \lat)$. The result then follows from Proposition~\ref{prop:DGS_approx_CVP}, which says that, except with probability $\exp(-\Omega(n))$ a sample from $D_{\lat' - \vec{y} - \vec{t}, s_i/2^{\ell/2}}$ will be a $(1+\exp(-\Omega(n/\log^2n)))$-approximate shortest vector in $\lat' - \vec{y} - \vec{t}$. 
\end{proof}

\end{document}